\documentclass[nohyperref]{article}

\usepackage{microtype}
\usepackage{graphicx}
\usepackage{booktabs} %
\usepackage{enumitem}

\usepackage{hyperref}

\usepackage{amsmath}
\usepackage{amssymb}
\usepackage{mathtools}
\usepackage{amsthm}
\usepackage{bm}
\usepackage{bbm}
\usepackage{etoolbox}
\usepackage{multirow}
\usepackage{subcaption}
\usepackage{float}

\newtoggle{arxiv}

\toggletrue{arxiv}

\iftoggle{arxiv}{
  \usepackage[square,numbers,sort]{natbib}
  \setlength{\textwidth}{6.5in}
  \setlength{\textheight}{9in}
  \setlength{\oddsidemargin}{0in}
  \setlength{\evensidemargin}{0in}
  \setlength{\topmargin}{-0.5in}
  \newlength{\defbaselineskip}
  \setlength{\defbaselineskip}{\baselineskip}
  \setlength{\marginparwidth}{0.8in}
  \setlength{\parskip}{4pt}%
    \setlength{\parindent}{0pt}%
  \usepackage[dvipsnames]{xcolor}         %
}{
  \usepackage{icml2022}
  \addtolength{\tabcolsep}{-2pt} %
  \def\setstretch#1{\renewcommand{\baselinestretch}{#1}}
  \setstretch{0.99}
  \addtolength{\parskip}{-1pt}
}

\usepackage[capitalize,noabbrev]{cleveref}

\usepackage{pifont}%
\newcommand{\cmark}{\ding{51}}%
\newcommand{\xmark}{\ding{55}}%

\newcommand{\A}{{A}}
\newcommand{\B}{{B}}
\newcommand{\C}{{C}}

\newcommand{\Abar}{\overline{A}}

\newcommand{\Kbar}{\overline{K}}

\newcommand{\dt}{\Delta}

\newcommand{\diag}{\Lambda}
\newcommand{\p}{p}

\newcommand{\D}{{D}}

\newcommand{\name}{\textsc{SaShiMi}}

\theoremstyle{plain}
\newtheorem{theorem}{Theorem}[section]
\newtheorem{proposition}[theorem]{Proposition}

\theoremstyle{definition}
\newtheorem{definition}[theorem]{Definition}

\theoremstyle{remark}

\usepackage[textsize=tiny]{todonotes}

\iftoggle{arxiv}{
  \title{It's Raw! Audio Generation with State-Space Models}
  \usepackage{authblk}
  \author[]{Karan Goel}
  \author[]{Albert Gu}
  \author[]{Chris Donahue}
  \author[]{Christopher R{\'e}}
  \affil[]{Department of Computer Science, Stanford University}
  \affil[]{{\texttt{ \{kgoel,albertgu,cdonahue,chrismre\}@cs.stanford.edu}}}
  \bibliographystyle{plainnat}
}
{
  \bibliographystyle{icml2022}
}

\begin{document}

\iftoggle{arxiv}{
  \date{}
  \maketitle
}{}

\begin{abstract}
Developing architectures suitable for modeling raw audio is a challenging problem due to the high sampling rates of audio waveforms.
Standard sequence modeling approaches like RNNs and CNNs have previously been tailored to fit the demands of audio,
but the resultant architectures make undesirable computational tradeoffs and struggle to model waveforms effectively. 
We propose \name{}, a new multi-scale architecture for waveform modeling built around the recently introduced S4 model for long sequence modeling.
We identify that S4 can be unstable during autoregressive generation, and provide a simple improvement to its parameterization by drawing connections to Hurwitz matrices.
\name{} yields state-of-the-art performance for unconditional waveform generation in the autoregressive setting. 
Additionally, \name{} improves non-autoregressive generation performance when used as the backbone architecture for a diffusion model. 
Compared to prior architectures in the autoregressive generation setting,
\name{} generates piano and speech waveforms which humans find more musical and coherent respectively, e.g.\ $2\times$ better mean opinion scores than WaveNet on an unconditional speech generation task. On a music generation task, 
\name{} 
outperforms WaveNet on density estimation and speed at both training and inference 
even when using $3\times$ fewer parameters. 
\iftoggle{arxiv}{Code can be found at \url{https://github.com/HazyResearch/state-spaces} and samples at \url{https://hazyresearch.stanford.edu/sashimi-examples}.}{}

\end{abstract}

\begin{figure*}[!t]
\begin{minipage}{1.0\linewidth}%
    \centering
    \includegraphics[width=\linewidth]{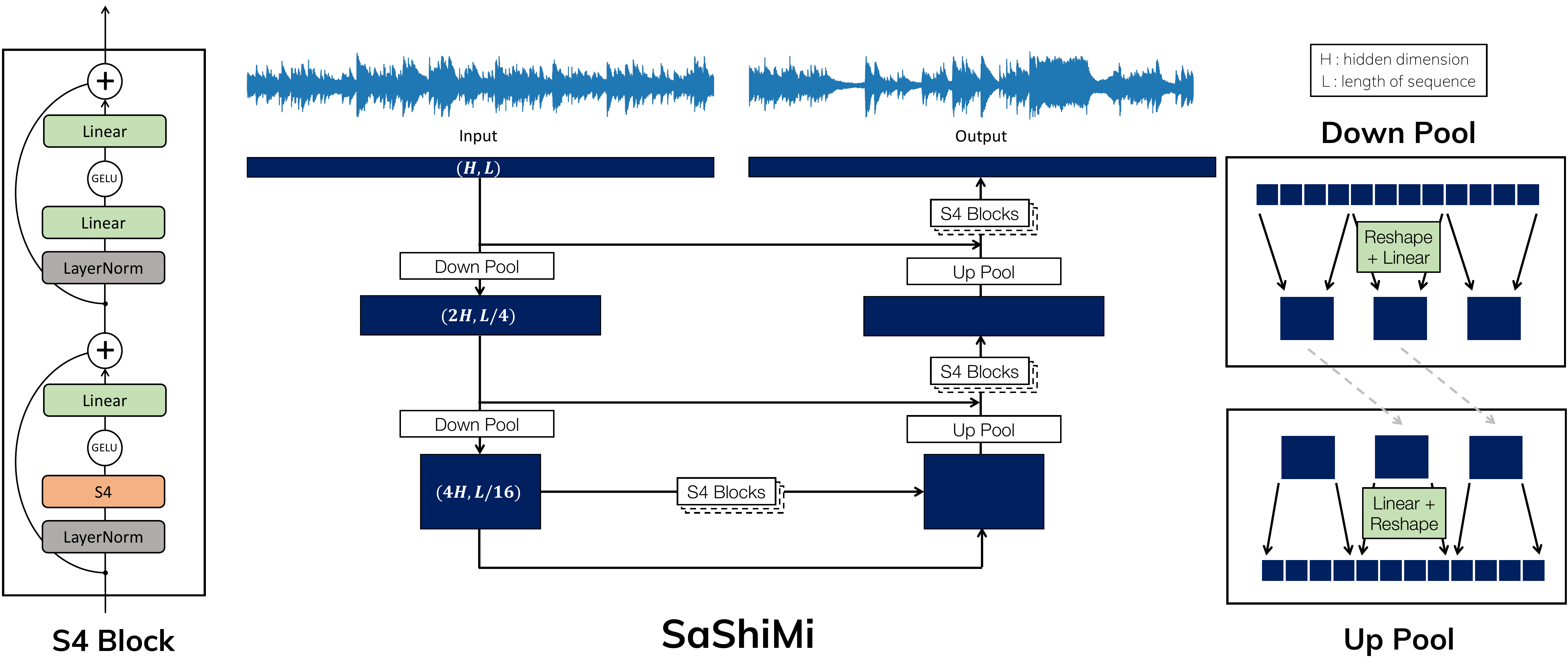}
\end{minipage}
\caption{
  \name{} consists of a simple repeated block combined with a multiscale architecture.
  (\emph{Left}) The basic S4 block is composed of an S4 layer combined with standard pointwise linear functions, non-linearities, and residual connections.
  (\emph{Center}) Dark blue rectangles illustrate the shape of inputs.
  The input is progressively transformed to shorter and wider sequences through pooling layers, then transformed back with stacks of S4 blocks.
  Longer range residual connections are included to help propagate signal through the network.
  (\emph{Right}) Pooling layers are position-wise linear transformations with shifts to ensure causality.
}
\label{fig:architecture}
\end{figure*}

\section{Introduction}
\label{sec:intro}

Generative modeling of raw audio \emph{waveforms} is a challenging frontier for machine learning due to their high-dimensionality---waveforms contain tens of thousands of timesteps per second and exhibit long-range behavior at multiple timescales. 
A key problem is developing architectures for modeling waveforms with the following properties:
\iftoggle{arxiv}{\begin{enumerate}}{\begin{enumerate}[leftmargin=*,itemsep=0pt]}
    \item {\bf Globally coherent} generation, which requires modeling unbounded contexts with long-range dependencies. 
    \item {\bf Computational efficiency} through parallel training, and fast autoregressive and non-autoregressive inference.
    \item {\bf Sample efficiency} through a model with inductive biases well suited to high-rate waveform data.
\end{enumerate}

Among the many training methods for waveform generation, autoregressive (AR) modeling is a fundamentally important approach. 
AR models learn the distribution of future variables conditioned on past observations, and are central to 
recent advances in machine learning for 
language and image generation~\citep{brown2020language,ramesh2021zero,bommasani2021opportunities}. 
With AR models, computing the exact likelihood is tractable, which makes them simple to train, and lends them to applications such as lossless compression~\citep{kleijn2018wavenet} and posterior sampling~\citep{jayaram2021parallel}. 
When generating, they can condition on arbitrary amounts of past context to sample sequences of unbounded length---potentially even longer than contexts observed during training. 
Moreover, architectural developments in AR waveform modeling can have a cascading effect on audio generation more broadly. 
For example, WaveNet---the earliest such architecture~\citep{oord2016wavenet}---remains a central component of 
state-of-the-art approaches for text-to-speech~(TTS)~\citep{li2019neural},
unconditional generation~\citep{lakhotia2021generative}, 
and non-autoregressive (non-AR) generation~\citep{kong2021diffwave}. 

Despite notable progress in AR modeling of (relatively) short sequences found in domains such as natural language (e.g.~$1$K tokens), 
it is still an open challenge to develop architectures that are effective for the much longer sequence lengths of audio waveforms (e.g.~$1$M samples). 
Past attempts have tailored standard sequence modeling approaches like CNNs~\citep{oord2016wavenet}, RNNs~\citep{mehri2017samplernn}, and Transformers~\citep{child2019generating} to fit the demands of AR waveform modeling, 
but these approaches have limitations. 
For example, RNNs lack computational efficiency because they cannot be parallelized during training, while CNNs cannot achieve global coherence because they are fundamentally constrained by the size of their receptive field.

We introduce {\bf \name{}}, a new architecture for modeling waveforms that yields state-of-the-art performance on unconditional audio generation benchmarks in both the AR and non-AR settings. 
\name{} is designed around recently developed deep state space models (SSM), specifically S4~\citep{gu2022efficiently}.
SSMs have a number of key features that make them ideal for modeling raw audio data. Concretely, S4: %

\iftoggle{arxiv}{\begin{enumerate}}{\begin{enumerate}[leftmargin=*,itemsep=0pt]}
    \item Incorporates a principled approach to modeling long range dependencies with strong results on long sequence modeling, including raw audio classification.
    \item Can be computed either as a CNN for efficient parallel training, or an RNN for fast autoregressive generation.%
    \item Is implicitly a continuous-time model, making it well-suited to signals like waveforms.
\end{enumerate}

To realize these benefits of SSMs inside \name{}, we make $3$ technical contributions.
First, we observe that while stable to train, S4's recurrent representation cannot be used for autoregressive generation due to numerical instability.
We identify the source of the instability using classical state space theory,
which states that SSMs are stable when the state matrix is Hurwitz,
which is not enforced by the S4 parameterization.
We provide a simple improvement to the S4 parameterization that theoretically ensures stability.

Second, \name{} incorporates pooling layers between blocks of residual S4 layers to capture hierarchical information across multiple resolutions. %
This is a common technique in neural network architectures such as standard CNNs and multi-scale RNNs, and provides empirical improvements in both performance and computational efficiency over isotropic stacked S4 layers.

Third, while S4 is a causal (unidirectional) model suitable for AR modeling, we provide a simple bidirectional relaxation to flexibly incorporate it in non-AR architectures.
This enables it to better take advantage of the available global context in non-AR settings.

For AR modeling in audio domains with unbounded sequence lengths (e.g.~music), 
\name{} can train on much longer contexts than existing methods including WaveNet (sequences of length $128$K vs $4$K),
while simultaneously having better test likelihood, faster training and inference, and fewer parameters.
\name{} outperforms existing AR methods in modeling the data ($>0.15$ bits better negative log-likelihoods), with substantial improvements ($+0.4$ points) in the musicality of long generated samples ($16$s) as measured by mean opinion scores.
In unconditional speech generation, 
\name{} achieves superior global coherence compared to previous AR models 
on the difficult SC09 dataset both quantitatively ($80\%$ higher inception score) and qualitatively ($2\times$ higher audio quality and digit intelligibility opinion scores by human evaluators).

Finally, we validate that \name{} is a versatile backbone for non-AR architectures.
Replacing the WaveNet backbone with \name{} in the state-of-the-art diffusion model DiffWave improves its quality, sample efficiency, and robustness to hyperparameters with no additional tuning.

\paragraph{Our Contributions.}

The central contribution of this paper is showing that deep neural networks using SSMs are a strong alternative to conventional architectures for modeling audio waveforms, with favorable tradeoffs in training speed, generation speed, sample efficiency, and audio quality.
\iftoggle{arxiv}{\begin{itemize}}{\begin{itemize}[leftmargin=*,itemsep=0pt]}
 \item We technically improve the parameterization of S4, ensuring its stability when switching into recurrent mode at generation time.
 \item We introduce \name{}, an SSM-based architecture with high efficiency and performance for unconditional AR modeling of music and speech waveforms.
 \item We show that \name{} is easily incorporated into other deep generative models to improve their performance.
\end{itemize}

\section{Related Work}
\label{sec:related}

This work focuses primarily on the task of generating raw audio waveforms without conditioning information.
Most past work on waveform generation involves conditioning on localized intermediate representations like
spectrograms~\citep{shen2018natural,kumar2019melgan,prenger2019waveglow},
linguistic features~\citep{oord2016wavenet,kalchbrenner2018efficient,binkowski2020high},
or discrete audio codes~\citep{oord2017neural,dieleman2018challenge,dhariwal2020jukebox,lakhotia2021generative}.
Such intermediaries provide copious information about the underlying content of a waveform, enabling generative models to produce globally-coherent waveforms while only modeling local structure.

In contrast, modeling waveforms in an unconditional fashion requires learning both local and global structure with a single model, and is thus more challenging.
Past work in this setting can be categorized into
AR approaches~\citep{oord2016wavenet,mehri2017samplernn,child2019generating}, where audio samples are generated one at a time given previous audio samples,
and non-AR approaches~\citep{donahue2019adversarial,kong2021diffwave}, where entire waveforms are generated in a single pass.
While non-AR approaches tend to generate waveforms more efficiently,
AR approaches have two key advantages.
First, unlike non-AR approaches, they can generate waveforms of unbounded length.
Second, they can tractably compute exact likelihoods,
allowing them to be used for compression~\citep{kleijn2018wavenet} and posterior sampling~\citep{jayaram2021parallel}.

In addition to these two advantages, new architectures for AR modeling of audio have the potential to bring about a cascade of improvements in audio generation more broadly.
For example, while the WaveNet architecture was originally developed for AR modeling (in both conditional and unconditional settings),
it has since become a fundamental piece of infrastructure in numerous audio generation systems.
For instance,
WaveNet is
commonly used to \emph{vocode}
intermediaries such as spectrograms~\citep{shen2018natural} or discrete audio codes~\citep{oord2017neural} into waveforms, often in the context of text-to-speech (TTS) systems.
Additionally, it serves as the backbone for several families of non-AR generative models of audio in both the conditional and unconditional settings:
\iftoggle{arxiv}{\begin{enumerate}[label=(\roman*),leftmargin=*]%
}{\begin{enumerate}[label=(\roman*),leftmargin=*,nolistsep]%
}
  \item Distillation: Parallel WaveNet~\citep{oord2018parallel} and ClariNet~\citep{ping2019clarinet} distill parallelizable flow models from a teacher WaveNet model.
  \item Likelihood-based flow models: WaveFlow~\citep{ping2020waveflow}, WaveGlow~\citep{prenger2019waveglow}, and FloWaveNet~\citep{kim2019flowavenet} all use WaveNet as a core component of reversible flow architectures.
  \item Autoencoders: WaveNet Autoencoder~\citep{engel2017neural} and WaveVAE~\citep{peng2020non}, which use WaveNets in their encoders.
  \item Generative adversarial networks (GAN): Parallel WaveGAN~\citep{yamamoto2020parallel} and GAN-TTS~\citep{binkowski2020high}, which use WaveNets in their discriminators.
  \item Diffusion probabilistic models: WaveGrad~\cite{chen2021wavegrad} and DiffWave~\cite{kong2021diffwave} learn a reversible noise diffusion process on top of dilated convolutional architectures.
\end{enumerate}
In particular, we point out that DiffWave represents the state-of-the-art for unconditional waveform generation, and incorporates
WaveNet as a black box. %

Despite its prevalence,
WaveNet is unable to model long-term structure beyond the length of its receptive field (up to $3$s), and in practice, may even fail to leverage available information beyond a few tens of milliseconds~\citep{shen2018natural}.
Hence,
we develop an alternative to WaveNet which can leverage unbounded context.
We focus primarily on evaluating our proposed architecture \name{} in the fundamental AR setting,
and additionally demonstrate that, like WaveNet,
\name{} can also transfer to non-AR settings.

\section{Background}
We provide relevant background on autoregressive waveform modeling in~\cref{sec:autoregressive}, state-space models in~\cref{sec:ssm} and the recent S4 model in~\cref{sec:s4}, before introducing \name{} in~\cref{sec:method}.
\label{sec:background}

\subsection{Autoregressive Modeling of Audio}
\label{sec:autoregressive}
Given a distribution over waveforms \( x = (x_0, \dots, x_{T-1}) \), autoregressive generative models model the joint distribution as the factorized product of conditional probabilities
\begin{align*}
  p(x) = \prod_{t=0}^{T-1} p(x_t | x_0, \dots, x_{t-1})
  .
\end{align*}

Autoregressive models have two basic modes:

\textbf{Training}: Given a sequence of samples \( x_0, \dots, x_{T-1} \),
maximize the likelihood
\[ p(x_0, \dots, x_{T-1}) = \sum_{i=0}^{T-1} p(x_i | x_0, \dots, x_{i-1}) = \sum_{i=0}^{T-1} \ell(y_i, x_{i+1}) \]
where \( \ell \) is the cross-entropy loss function.

\textbf{Inference (Generation)}: Given \( x_0, \dots, x_{t-1} \) as context, sample from the distribution represented by \( y_{t-1} = p(x_t \mid x_0, \dots, x_{t-1})\) to produce the next sample \( x_t \).

We remark that by the training mode, autoregressive models are equivalent to \emph{causal sequence-to-sequence maps} \( x_0, \dots, x_{T-1} \mapsto y_0, \dots, y_{T-1} \),
where \( x_k \) are input samples to model and \( y_k \) represents the model's guess of \( p_(x_{k+1} \mid x_0, \dots, x_k) \). %
For example, when modeling a sequence of categorical inputs over \( k \) classes, typically \( x_k \in \mathbbm{R}^d \) are embeddings of the classes and \( y_k \in \mathbbm{R}^k \) represents a categorical distribution over the classes.

The most popular models for autoregressive audio modeling are based on CNNs and RNNs,
which have different tradeoffs during training and inference.
A CNN layer computes a convolution with a parameterized kernel
\begin{equation}
  \label{eq:cnn}
  K = (k_0, \dots, k_{w-1}) \qquad  y = K \ast x
\end{equation}
where \( w \) is the width of the kernel.
The \emph{receptive field} or \emph{context size} of a CNN is the sum of the widths of its kernels over all its layers.
In other words, modeling a context of size \( T \) requires learning a number of parameters proportional to \( T \).
This is problematic in domains such as audio which require very large contexts.

A variant of CNNs particularly popular for modeling audio is the \emph{dilated convolution} (DCNN) popularized by WaveNet~\citep{oord2016wavenet},
where each kernel \( K \) is non-zero only at its endpoints.
By choosing kernel widths carefully, such as in increasing powers of \( 2 \), a DCNN can model larger contexts than vanilla CNNs.

RNNs such as SampleRNN~\citep{mehri2017samplernn} maintain a hidden state \( h_t \) that is sequentially computed from the previous state and current input,
and models the output as a function of the hidden state
\begin{equation}
  \label{eq:rnn}
  h_t = f(h_{t-1}, x_t) \qquad
  y_t = g(h_t)
\end{equation}
The function \( f \) is also known as an RNN cell,
such as the popular LSTM~\citep{lstm}.

CNNs and RNNs have efficiency tradeoffs as autoregressive models.
CNNs are \emph{parallelizable}: given an input sequence \( x_0, \dots, x_{T-1} \),
they can compute all \( y_k \) at once, making them efficient during training. %
However, they become awkward at inference time when only the output at a single timestep \( y_t \) is needed. Autoregressive stepping requires specialized caching implementations that have higher complexity requirements than RNNs.

On the other hand, RNNs are \emph{stateful}: The entire context \( x_0, \dots, x_{t} \) is summarized into the hidden state \( h_{t} \). %
This makes them efficient at inference, requiring only constant time and space to generate the next hidden state and output.
However, this inherent sequentiality leads to slow training and optimization difficulties (the vanishing gradient problem~\citep{hochreiter2001gradient,pascanu2013difficulty}).

\subsection{State Space Models}
\label{sec:ssm}

A recent class of deep neural networks was developed that have properties of both CNNs and RNNs.
The state space model (SSM) is defined in continuous time by the equations
\begin{equation}
  \label{eq:ssm}
  \begin{aligned}
    h'(t) &= \A h(t) + \B x(t) \\
    y(t)  &= \C h(t) + \D x(t)
  \end{aligned}
  .
\end{equation}

To operate on discrete-time sequences sampled with a step size of \( \dt \), SSMs can be computed with the recurrence
\begin{align}
  \label{eq:ssm-recurrence}
  h_{k} &= \overline{A} h_{k-1} + \overline{B} x_k
  \qquad y_k = \overline{C} h_k + \overline{D} x_k
  \\
  \label{eq:ssm-discretization}
  \overline{A} &= (I - \dt/2 \cdot A)^{-1}(I + \dt/2 \cdot A)
\end{align}
where \(\overline{A}\) is the \emph{discretized state matrix} and \( \overline{B}, \dots \) have similar formulas.
Eq.~\eqref{eq:ssm-recurrence} is equivalent to the convolution
\begin{equation}
  \label{eq:ssm-convolution}
\begin{aligned}
  \overline{K} =
  (\overline{C}\overline{B}, \overline{C}\overline{A}\overline{B}, \overline{C}\overline{A}^2\overline{B})
  \qquad
  y = \overline{K} \ast x %
  .
\end{aligned}
\end{equation}

SSMs can be viewed as particular instantiations of CNNs and RNNs that inherit their efficiency at both training and inference and overcome their limitations.
As an RNN, \eqref{eq:ssm-recurrence} is a special case of \eqref{eq:rnn} where \( f \)  and \( g \) are linear, giving it much simpler structure that avoids the optimization issues found in RNNs.
As a CNN, \eqref{eq:ssm-convolution} is a special case of \eqref{eq:cnn} with an unbounded convolution kernel, overcoming the context size limitations of vanilla CNNs.

\subsection{S4}
\label{sec:s4}

S4 is a particular instantiation of SSM that parameterizes  \( A \) as a \emph{diagonal plus low-rank} (DPLR) matrix,
\( A = \Lambda + p q^* \) \citep{gu2022efficiently}.
This parameterization has two key properties.
First, this is a structured representation that allows faster computation---S4 uses a special algorithm to compute the convolution kernel
\( \Kbar \) \eqref{eq:ssm-convolution} very quickly. %
Second, this parameterization includes certain special matrices called HiPPO matrices \citep{gu2020hippo}, which theoretically and empirically allow the SSM to capture long-range dependencies better.
In particular, HiPPO specifies a special equation \( h'(t) = Ah(t) + Bx(t) \) with closed formulas for \( A \) and \( B \).
This particular \( A \) matrix can be written in DPLR form, and S4 initializes its \( A \) and \( B \) matrices to these.

\section{Model}
\label{sec:method}

\begin{table*}[t!]
    \centering
    \caption{Summary of music and speech datasets used for unconditional AR generation experiments.}
    \sc
    \resizebox{\linewidth}{!}{
        \begin{tabular}{lllllll}
            \toprule
            Category & Dataset    & Total Duration & Chunk Length & Sampling Rate & Quantization   & Splits (train-val-test)        \\
            \midrule
            Music    & Beethoven  & $10$ hours     & $8$s  & $16$kHz       & $8$-bit linear & \citet{mehri2017samplernn} \\
            Music    & YouTubeMix & $4$ hours      & $8$s      & $16$kHz       & $8$-bit mu-law & $88\%-6\%-6\%$  \\
            Speech   & SC09       & $5.3$ hours    & $1$s            & $16$kHz       & $8$-bit mu-law & \citet{Warden2018SpeechCA}       \\
            \bottomrule
        \end{tabular}
    }
    \label{tab:datasets}
\end{table*}

\name{} consists of two main components.
First, S4 layers are the core component of our neural network architecture, to capture long context while being fast at both training and inference.
We provide a simple improvement to S4 that addresses instability at generation time~(\cref{sec:stable}).
Second, \name{} connects stacks of S4 layers together in a simple multi-scale architecture (\cref{sec:architecture}).

\subsection{Stabilizing S4 for Recurrence}
\label{sec:stable}

We use S4's representation and algorithm as a black box, with one technical improvement: we use the parameterization \( \Lambda - pp^* \) instead of \( \Lambda + pq^* \).
This amounts to essentially tying the parameters \( p \) and \( q \) (and reversing a sign).

To justify our parameterization, we first note that it still satisfies the main properties of S4's representation (\cref{sec:s4}).
 First, this is a special case of a DPLR matrix, and can still use S4's algorithm for fast computation.
 Moreover, we show that the HiPPO matrices still satisfy this more restricted structure; in other words, we can still use the same initialization which is important to S4's performance.
\begin{proposition}%
  \label{prop:hippo}
  All three HiPPO matrices from \citep{gu2020hippo} are unitarily equivalent to a matrix of the form
  \(
    \A = \diag - \p \p^*
  \)
  for diagonal \( \diag \) and \( \p \in \mathbbm{R}^{N \times r} \) for \( r=1 \) or \( r=2 \).
  Furthermore, all entries of \( \diag \) have real part \( 0 \) (for HiPPO-LegT and HiPPO-LagT) or \( -\frac{1}{2} \) (for HiPPO-LegS).
\end{proposition}

 Next, we discuss how this parameterization makes S4 stable.
 The high-level idea is that stability of SSMs involves the spectrum of the state matrix \( \A \), which is more easily controlled because \( -pp^* \) is a negative semidefinite matrix (i.e., we know the signs of its spectrum).

 \begin{definition}%
   A \emph{Hurwitz matrix} \( \A \) is one where every eigenvalue has negative real part.
 \end{definition}
 Hurwitz matrices are also called stable matrices, because they imply that the SSM~\eqref{eq:ssm} is asymptotically stable.
 In the context of discrete time SSMs, we can easily see why \( \A \) needs to be a Hurwitz matrix from first principles with the following simple observations.

 First, unrolling the RNN mode (equation \eqref{eq:ssm-recurrence}) involves powering up \( \Abar \) repeatedly, which is stable if and only if all eigenvalues of \( \Abar \) lie inside or on the unit disk.
 Second, the transformation \eqref{eq:ssm-discretization} maps the complex left half plane (i.e. negative real part) to the complex unit disk.
 Therefore computing the RNN mode of an SSM (e.g. in order to generate autoregressively) requires \( \A \) to be a Hurwitz matrix.

However, controlling the spectrum of a general DPLR matrix is difficult;
empirically, we found that S4 matrices generally became non-Hurwitz after training.
We remark that this stability issue only arises when using S4 during autoregressive generation, because S4's convolutional mode during training does not involve powering up \( \Abar \) and thus does not require a Hurwitz matrix.
Our reparameterization makes controlling the spectrum of \( \Abar \) easier.

\begin{proposition}%
  \label{prop:stable}
  A matrix \(  A = \Lambda - p p^* \) is Hurwitz if all entries of \( \Lambda \) have negative real part.
\end{proposition}
\begin{proof}%
  We first observe that if \( \A + \A^* \) is negative semidefinite (NSD), then \( \A \) is Hurwitz.
  This follows because \( 0 > v^* (\A+\A^*) v = (v^*\A v) + (v^* \A v)^* = 2\mathfrak{Re}(v^*\A v) = 2\lambda \) for any (unit length) eigenpair \( (\lambda, v) \) of \( \A \).

  Next, note that the condition implies that \( \diag + \diag^* \) is NSD (it is a real diagonal matrix with non-positive entries).
  Since the matrix \( - \p \p^* \) is also NSD, then so is \( \A + \A^* \).
\end{proof}

\cref{prop:stable} implies that with our tied reparameterization of S4, controlling the spectrum of the learned \( \A \) matrix becomes simply controlling the the diagonal portion \( \diag \).
This is a far easier problem than controlling a general DPLR matrix, and can be enforced by regularization or reparameteration (e.g. run its entries through an \( \exp \) function).
In practice, we found that not restricting \( \diag \) and letting it learn freely led to stable trained solutions.

\subsection{\name{} Architecture}
\label{sec:architecture}

\cref{fig:architecture} illustrates the complete \name{} architecture.

\textbf{S4 Block.}
\name{} is built around repeated deep neural network blocks containing our modified S4 layers,
following the same original S4 model.
Compared to \citet{gu2022efficiently}, we add additional pointwise linear layers after the S4 layer in the style of the \emph{feed-forward network} in Transformers or the \emph{inverted bottleneck layer} in CNNs~\citep{liu2022convnet}.
Model details are in \cref{sec:model-details}.

\textbf{Multi-scale Architecture.}
\name{} uses a simple architecture for autoregressive generation that consolidates information from the raw input signal at multiple resolutions.
The \name{} architecture consists of multiple tiers, with each tier composed of a stack of residual S4 blocks.
The top tier processes the raw audio waveform at its original sampling rate, while lower tiers process downsampled versions of the input signal.
The output of lower tiers is upsampled and combined with the input to the tier above it in order to provide a stronger conditioning signal.
This architecture is inspired by related neural network architectures for AR modeling that incorporate multi-scale characteristics such as SampleRNN and PixelCNN++~\citep{salimans2017pixelcnn++}.

The pooling is accomplished by simple reshaping and linear operations.
Concretely, an input sequence \( x \in \mathbbm{R}^{T \times H} \) with context length \( T \) and hidden dimension size \( H \) is transformed through these shapes:
\iftoggle{arxiv}
{
\begin{align*}
  (\textbf{Down-pool}) \, (T, H) \xrightarrow{\text{reshape}} (T/p, p \cdot H) \xrightarrow{\text{linear}} (T / p, q \cdot H) \\
  (\textbf{Up-pool}) \, (T, H) \xrightarrow{\text{linear}} (T, p \cdot H / q) \xrightarrow{\text{reshape}} (T \cdot p, H / q).
\end{align*}}
{\small
\begin{align*}
  (\textbf{Down-pool}) & (T, H) \xrightarrow{\text{reshape}} (T/p, p \cdot H) \xrightarrow{\text{linear}} (T / p, q \cdot H) \\
  (\textbf{Up-pool}) & (T, H) \xrightarrow{\text{linear}} (T, p \cdot H / q) \xrightarrow{\text{reshape}} (T \cdot p, H / q).
\end{align*}}

\normalsize
Here, \( p \) is the \emph{pooling factor} and \( q \) is an \emph{expansion factor} that increases the hidden dimension while pooling.
In our experiments, we always fix \( p=4, q=2 \)
and use a total of just two pooling layers (three tiers).

We additionally note that in AR settings, the up-pooling layers must be shifted by a time step to ensure causality.

\textbf{Bidirectional S4.}
Like RNNs, SSMs are causal with an innate time dimension (equation \eqref{eq:ssm}).
For non-autoregressive tasks, we consider a simple variant of S4 that is bidirectional.
We simply pass the input sequence through an S4 layer, and also reverse it and pass it through an independent second S4 layer.
These outputs are concatenated and passed through a positionwise linear layer as in the standard S4 block.
\begin{align*}
  y = \mathsf{Linear}(\mathsf{Concat}(\mathsf{S4}(x), \mathsf{rev}(\mathsf{S4}(\mathsf{rev}(x)))))
\end{align*}
We show that bidirectional S4 outperforms causal S4 when autoregression is not required (\cref{sec:diffwave}).

\begin{table}[!t]
    \centering
    \caption{
      Results on AR modeling of Beethoven, a benchmark task from~\citet{mehri2017samplernn}---\name{} outperforms all baselines while training faster.
    }
    \sc
    
    \resizebox{\iftoggle{arxiv}{0.6\linewidth}{\linewidth}}{!}{
        \begin{tabular}{@{}lllll@{}}
            \toprule
            Model & Context       & NLL & @$200$K steps & @$10$ hours   \\
            \midrule
            SampleRNN$^*$          & $1024$        & $1.076$                  & $-$          & $-$          \\
            WaveNet$^*$            & $4092$        & $1.464$                  & $-$          & $-$          \\
            \midrule
            SampleRNN$^\dagger$    & $1024$        & $1.125$                  & $1.125$      & $1.125$      \\ %
            WaveNet$^\dagger$      & $4092$        & $1.032$                  & $1.088$      & $1.352$      \\ %
            \name{}                & $\bm{128000}$ & $\bm{0.946}$             & $\bm{1.007}$ & $\bm{1.095}$ \\ %
            \bottomrule
            \multicolumn{5}{l}{\footnotesize{$^*$Reported in \citet{mehri2017samplernn} \qquad $^\dagger$Our replication}}
        \end{tabular}
    }
    \label{tab:beethoven}
\end{table}

\begin{table}[!t]
  \centering
  \small
  \caption{
    Effect of context length on the performance of \name{} on Beethoven, controlling for computation and sample efficiency. \name{} is able to leverage information from longer contexts.
  }
  \sc
  \resizebox{\iftoggle{arxiv}{0.6\linewidth}{\linewidth}}{!}{
    \begin{tabular}{llll}
      \toprule
      \multirow{2}{*}{Context Size} &
      \multirow{2}{*}{Batch Size}   & \multicolumn{2}{c}{NLL} \\
      \cmidrule{3-4}
                                    &                          & $200$K steps & $10$ hours   \\
      \midrule
      $1$ second                    & $8$                      & $1.364$      & $1.433$      \\
      $2$ seconds                   & $4$                      & $1.229$      & $1.298$      \\
      $4$ seconds                   & $2$                      & $1.120$      & $1.234$      \\
      $8$ seconds                   & $1$                      & $\bm{1.007}$ & $\bm{1.095}$ \\
      \bottomrule
    \end{tabular}
  }
  \label{tab:beethoven-context-length}
\end{table}

\section{Experiments}

We evaluate \name{} on several benchmark audio generation and unconditional speech generation tasks in both AR and non-AR settings,
validating that \name{} generates more globally coherent waveforms than baselines while having higher computational and sample efficiency.

{\bf Baselines.} We compare \name{} to the leading
AR models for unconditional waveform generation, SampleRNN and WaveNet.
In \cref{sec:diffwave}, we show that \name{} can also improve non-AR models. 

{\bf Datasets.} We evaluate \name{} on datasets spanning music and speech generation (\cref{tab:datasets}).
\begin{itemize}[leftmargin=*,itemsep=0pt]
  \item {\bf Beethoven.} A benchmark music dataset~\cite{mehri2017samplernn}, consisting of Beethoven's piano sonatas. %
  \item {\bf YouTubeMix.} Another piano music dataset~\cite{deepsound} with higher-quality recordings than Beethoven. %
  \item {\bf SC09.} A benchmark speech dataset~\cite{donahue2019adversarial}, consisting of $1$-second recordings of the digits ``zero'' through ``nine'' spoken by many different speakers. %
\end{itemize}

All datasets are quantized using $8$-bit quantization, either linear or $\mu$-law,
depending on prior work.
Each dataset is divided into non-overlapping chunks;
the SampleRNN baseline is trained using TBPTT, while WaveNet and \name{} are trained on entire chunks.
All models are trained to predict the negative log-likelihood (NLL) of individual audio samples; results are reported in base $2$,
also known as bits per byte (BPB) because of the one-byte-per-sample quantization.
All datasets were sampled at a rate of $16$kHz.
\cref{tab:datasets} summarizes characteristics of the datasets and processing.

\subsection{Unbounded Music Generation}
\label{sec:music}

Because music audio is not constrained in length,
AR models are a natural approach for music generation, 
since they can generate samples longer than the context windows they were trained on. 
We validate that \name{} can leverage longer contexts to perform
music waveform
generation more effectively than baseline AR methods.

We follow the setting of \citet{mehri2017samplernn} for the Beethoven dataset.
\cref{tab:beethoven} reports results found in prior work, as well as our reproductions.
In fact, our WaveNet baseline is much stronger than the one implemented in prior work.
\name{} substantially improves the test NLL by $0.09$ BPB compared to the best baseline.
\cref{tab:beethoven-context-length} ablates the context length used in training, showing that \name{} significantly benefits from seeing longer contexts,
and is able to effectively leverage extremely long contexts (over $100$k steps) when predicting next samples.

Next, we evaluate all baselines on YouTubeMix. Table~\ref{tab:youtube} shows that \name{} substantially outperforms SampleRNN and WaveNet on NLL. %
Following \citet{dieleman2018challenge} (protocol in \cref{sec:mos}), we
measured
mean opinion scores (MOS) for audio fidelity and musicality
for
$16$s samples generated by each method
(longer than the training context).
All methods have similar fidelity, but \name{} substantially improves musicality by around $0.40$ points, validating that it can generate long samples more coherently than other methods.

\begin{table}[t]
    \centering
    \caption{
        Negative log-likelihoods and mean opinion scores on YouTubeMix.
        As suggested by \citet{dieleman2018challenge}, we encourage readers to form their own opinions by referring to the sound examples in our supplementary material.
    }
    \label{tab:youtube}
    \sc
    \resizebox{\iftoggle{arxiv}{0.6\linewidth}{\linewidth}}{!}{
        \begin{tabular}{llll}
            \toprule
            Model     & Test NLL     & MOS (fidelity)       & MOS (musicality)     \\
            \midrule
            SampleRNN & $1.723$      & $\bm{2.98 \pm 0.08}$ & $1.82 \pm 0.08$      \\
            WaveNet   & $1.449$      & $2.91 \pm 0.08$      & $2.71 \pm 0.08$      \\
            \name     & $\bm{1.294}$ & $2.84 \pm 0.09$      & $\bm{3.11 \pm 0.09}$ \\
            \midrule
            Dataset   & -            & $3.76 \pm 0.08$      & $4.59 \pm 0.07$      \\
            \bottomrule
        \end{tabular}
    }
\end{table}

\begin{figure}[!t]
    \centering
    \caption{{\bf (Sample Efficiency)} Plot of validation NLL (in bits) vs. wall clock time (hours) on the SC09 dataset.}
    \includegraphics[width=\iftoggle{arxiv}{0.6\linewidth}{\linewidth}]{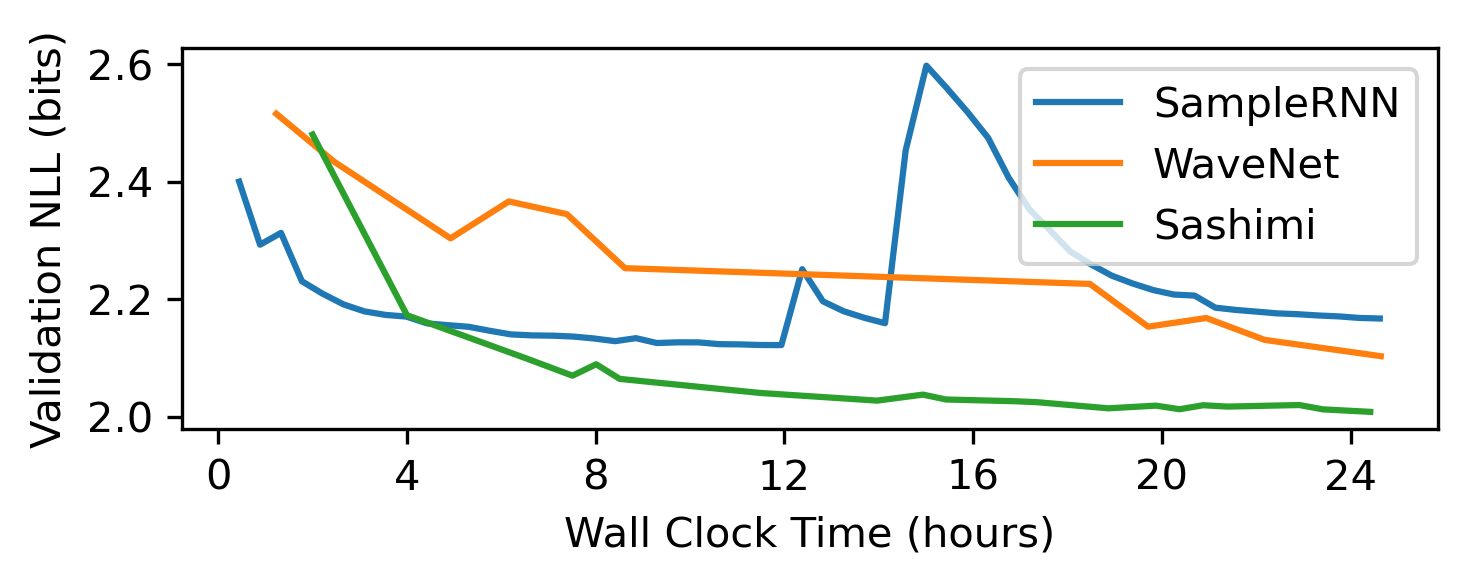}
    \label{fig:sc09-nll-vs-wallclock}
    \vspace*{-2em}
\end{figure}

\cref{fig:sc09-nll-vs-wallclock} shows that \name{} trains stably and more efficiently than baselines in wall clock time.
\cref{sec:experiments-additional}, \cref{fig:throughput} also analyzes the peak throughput of different AR models as a function of batch size.

\begin{table}[!h]
  \centering
  \caption{
    Architectural ablations and efficiency tradeoffs on YouTubeMix.
    (\emph{Top}) AR models and baselines at different sizes.
    (\emph{Bottom}) Ablating the pooling layers of \name{}.
  }
  \sc
  \resizebox{\iftoggle{arxiv}{0.75\linewidth}{\linewidth}}{!}{
    \begin{tabular}{lllll}
      \toprule
      Model              & NLL     & Time/epoch & Throughput        & Params   \\
      \midrule
      SampleRNN$-2$ tier & $1.762$ & $800$s     & $112$K samples/s  & $51.85$M \\
      SampleRNN$-3$ tier & $1.723$ & $850$s     & $116$K samples/s  & $35.03$M \\
      \midrule
      WaveNet$-512$      & $1.467$ & $1000$s    & $185$K samples/s  & $2.67$M  \\
      WaveNet$-1024$     & $1.449$ & $1435$s    & $182$K samples/s  & $4.24$M  \\
      \midrule
      \name$-2$ layers   & $1.446$ & $205$s     & $596$K  samples/s & $1.29$M  \\
      \name$-4$ layers   & $1.341$ & $340$s     & $316$K  samples/s & $2.21$M  \\
      \name$-6$ layers   & $1.315$ & $675$s     & $218$K  samples/s & $3.13$M  \\
      \name$-8$ layers   & $1.294$ & $875$s     & $129$K  samples/s & $4.05$M  \\
      \midrule
      \midrule
      Isotropic S4$-4$ layers & $1.429$ & $1900$s & $144$K samples/s & $2.83$M \\
      Isotropic S4$-8$ layers & $1.524$ & $3700$s & $72$K samples/s & $5.53$M \\
      \bottomrule
    \end{tabular}
  }
  \label{tab:computation-ablation}
\end{table}

\begin{table*}[t]
    \centering
    \caption{
      (\textbf{SC09}) Automated metrics and human opinion scores. (\emph{Top}) \name{} is the first AR model that can unconditionally generate high quality samples on this challenging dataset of fixed-length speech clips with highly variable characteristics.
      (\emph{Middle}) As a flexible architecture for general waveform modeling, \name{} sets a true state-of-the-art when combined with a recent diffusion probabilistic model.
    }
    \sc
    \resizebox{\linewidth}{!}{
        \begin{tabular}{lllllllllllll}
            \toprule
            \multirow{2}{*}{Model} & \multirow{2}{*}{Params} & \multirow{2}{*}{NLL} & \multirow{2}{*}{FID $\downarrow$} & \multirow{2}{*}{IS $\uparrow$} & \multirow{2}{*}{mIS $\uparrow$} & \multirow{2}{*}{AM $\downarrow$} & \multirow{2}{*}{\shortstack{Human (\( \kappa \)) \\ Agreement}}                      & \multicolumn{3}{c}{MOS}          \\
            \cmidrule{9-11}
                                   &                         &                      &                                   &                                &                                 &                                  &                                                   & Quality                          & Intelligibility                   & Diversity                        \\
            \midrule
            SampleRNN              & $35.0$M                 & $2.042$              & $8.96$                            & $1.71$                         & $3.02$                          & $1.76$                           & $0.321$                                           & $1.18 \pm 0.04$                  & $1.37 \pm 0.02$                   & $2.26 \pm 0.10$                  \\
            WaveNet                & $4.2$M                  & $1.925$              & $5.08$                            & $2.27$                         & $5.80$                          & $1.47$                           & $0.408$                                           & $1.59 \pm 0.06$                  & $1.72 \pm 0.03$                   & $2.70 \pm 0.11$                  \\
            \name{}                & $4.1$M                  & $\bm{1.891}$         & $\bm{1.99}$                       & $\bm{4.12}$                    & $\bm{24.57}$                    & $\bm{0.90}$                      & $\bm{0.832}$                                           & $\bm{3.29 \pm 0.07}$             & $\bm{3.53 \pm 0.04}$              & $\bm{3.26 \pm 0.09}$             \\
            \midrule
            WaveGAN                & $19.1$M                 & -                    & $2.03$                            & $4.90$                         & $36.10$                         & $0.80$                           & $0.840$                                           & $2.98 \pm 0.07$                  & $3.27 \pm 0.04$                   & $3.25 \pm 0.09$                  \\ %
            DiffWave               & $24.1$M                 & -                    & $1.92$                            & $5.26$                         & $51.21$                         & $0.68$                           & $0.917$                                           & $4.03 \pm 0.06$                  & $4.15 \pm 0.03$                   & $\bm{3.45 \pm 0.09}$                  \\
            ~~w/ \name{}           & $23.0$M                 & -                    & $\bm{1.42}$                       & $\bm{5.94}$                    & $\bm{69.17}$                    & $\bm{0.59}$                      & $\bm{0.953}$                                           & $\bm{4.20 \pm 0.06}$                  & $\bm{4.33 \pm 0.03}$                   & $3.28 \pm 0.11$                  \\
            \midrule
            Train                  & -                       & -                    & $0.00$                            & $8.56$                         & $292.5$                         & $0.16$                           & \multirow{2}{*}{$0.921$}                          & \multirow{2}{*}{$4.04 \pm 0.06$} & \multirow{2}{*}{$4.27 \pm 0.03$}  & \multirow{2}{*}{$3.59 \pm 0.09$} \\
            Test                   & -                       & -                    & $0.02$                            & $8.33$                         & $257.6$                         & $0.19$                           &                                                   &                                  &                                  \\
            \bottomrule
        \end{tabular}
    }
    \label{tab:sc09}
\end{table*}

\subsection{Model ablations: Slicing the \name{}}
\label{sec:ablations}

We validate our technical improvements and ablate \name{}'s architecture.

\textbf{Stabilizing S4.}
We consider how different parameterizations of S4's representation affect downstream performance (\cref{sec:stable}).
Recall that S4 uses a special matrix \( A = \Lambda + pq^* \) specified by HiPPO, which theoretically captures long-range dependencies (\cref{sec:s4}).
We ablate various parameterizations of a small \name{} model ($2$ layers, $500$ epochs on YouTubeMix).
Learning $A$ yields consistent improvements, but becomes unstable at generation.
Our reparameterization allows \( A \) to be learned while preserving stability, agreeing with the analysis in \cref{sec:stable}. A visual illustration of the spectral radii of the learned $\Abar$ in the new parameterization is provided in \cref{fig:s4-stability}.

\begin{table}[H]
    \centering
    \sc
    \small
    \resizebox{\iftoggle{arxiv}{0.5\linewidth}{0.9\linewidth}}{!}{
      \begin{tabular}{llll}
        \toprule
        Learned        & Frozen         & NLL     & Stable generation \\
        \midrule
        $-$            & $\diag + pq^*$ & $1.445$ & \cmark            \\
        $\diag + pq^*$ & $-$            & $1.420$ & \xmark            \\
        $\diag - pp^*$ & $-$            & $1.419$ & \cmark            \\
        \bottomrule
      \end{tabular}
    }
    \label{tab:s4-ablation}
\end{table}

\begin{figure}
    \centering
    \caption{\textbf{(S4 Stability)} Comparison of spectral radii for all $\Abar$ matrices in a SaShiMi model trained with different S4 parameterizations. The instability in the standard S4 parameterization is solved by our Hurwitz parameterization.}
    \includegraphics[width=\iftoggle{arxiv}{0.5\linewidth}{\linewidth}]{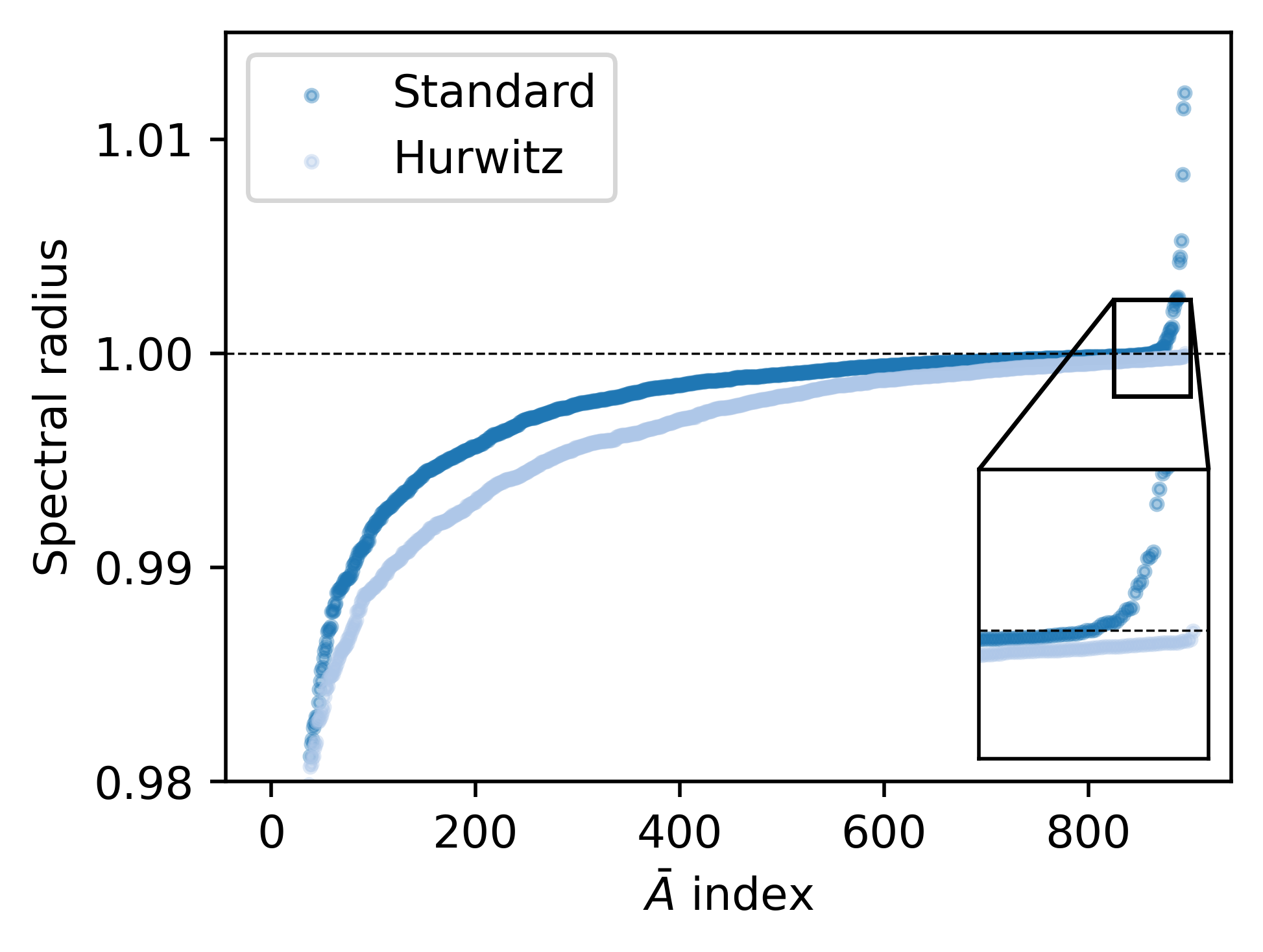}
    \label{fig:s4-stability}
\end{figure}

\textbf{Multi-scale architecture.}
We investigate the effect of \name{}'s architecture (\cref{sec:architecture}) against
isotropic S4 layers on YouTubeMix.
Controlling for parameter counts, adding pooling in \name{} leads to substantial improvements in computation and modeling (\cref{tab:computation-ablation}, Bottom).

\textbf{Efficiency tradeoffs.}
We ablate different sizes of the \name{} model on YouTubeMix to show its performance tradeoffs along different axes.

\cref{tab:computation-ablation} (Top) shows that a single \name{} model simultaneously outperforms all baselines on quality (NLL) and computation at both training and inference, with a model more than 3X smaller. 
Moreover, \name{} improves monotonically with depth, 
suggesting that quality can be further improved at the cost of additional computation.

\iftoggle{arxiv}
{
    \begin{table*}[t]
        \centering
        \caption{
          (\textbf{SC09 diffusion models}.)
          Beyond AR, \name{} can be flexibly combined with other generative modeling approaches, improving on the state-of-the-art DiffWave model by simply replacing the architecture.
          (\emph{Top}) A parameter-matched \name{} architecture with \emph{no tuning} outperforms the best DiffWave model.
          (\emph{Middle}) \name{} is consistently better than WaveNet at all stages of training; a model trained on half the samples matches the best DiffWave model.
          (\emph{Bottom}) The WaveNet backbone is extremely sensitive to architecture parameters such as size and dilation schedule; a small model fails to learn. We also ablate the bidirectional S4 layer, which outperforms the unidirectional one.
        }
        \sc
        \resizebox{\linewidth}{!}{
            \begin{tabular}{llllllllllll}
                \toprule
                            \multirow{2}{*}{Architecture} & \multirow{2}{*}{Params} & \multirow{2}{*}{Training Steps} & \multirow{2}{*}{FID $\downarrow$} & \multirow{2}{*}{IS $\uparrow$} & \multirow{2}{*}{mIS $\uparrow$} & \multirow{2}{*}{AM $\downarrow$} & \multirow{2}{*}{NDB $\downarrow$} & \multirow{2}{*}{\shortstack{Human (\( \kappa \))\\Agreement}} & \multicolumn{3}{c}{MOS} \\
                \cmidrule{10-12}
                &                         &                                   &                                &                                 &                                  &                                   & & & Quality                & Intelligibility     & Diversity          \\
    
                \midrule
                \name{}        & 23.0M      & 800k           & $1.42$              & $5.94$           & $69.17$            & $0.59$             & $0.88$              &   $0.953$                             &  $4.20 \pm 0.06$               & $4.33 \pm 0.03$ & $3.28 \pm 0.11$                      \\
                WaveNet        & 24.1M      & 1000k          & $1.92$           & $5.26$        & $51.21$        & $0.68$          & $0.88$           & $0.917$                        & $4.03 \pm 0.06$ & $4.15 \pm 0.03$        & $3.45 \pm 0.09$ \\
                \midrule
                \name{}        & 23.0M      & 500k           & $2.08$           & $5.68$        & $51.10$        & $0.66$          & $0.76$           & $0.923$                        & $3.99 \pm 0.06$ & $4.13 \pm 0.03$        & $3.38 \pm 0.10$ \\
                WaveNet        & 24.1M      & 500k           & $2.25$           & $4.68$        & $34.55$        & $0.80$           & $0.90$            & $0.848$                        & $3.53 \pm 0.07$ & $3.69 \pm 0.03$        & $3.30 \pm 0.08$ \\
                \midrule
                \name{} (uni.) & 7.1M       & 500k           & $2.70$              & $3.62$           & $17.96$            & $1.03$             & $0.90$              &  $0.829$                              & $3.08 \pm 0.07$                 & $3.29 \pm 0.04$ & $3.26 \pm 0.08$                       \\
                \name{}        & 7.5M       & 500k           & $1.70$              & $5.00$           & $40.27$        & $0.72$             & $0.90$              & $0.934$                               &  $3.83 \pm 0.07$               &     $4.00 \pm 0.03$ & $3.34 \pm 0.09$                   \\
                WaveNet        & 6.8M       & 500k           & $4.53$              & $2.80$           & $9.02$        & $1.30$             & $0.94$              &   $0.446$                             &  $1.85 \pm 0.08$               & $1.90 \pm 0.03$ & $3.03 \pm 0.10$                      \\
                \bottomrule
            \end{tabular}
        }
        \label{tab:sc09-diffusion}
    \end{table*}
}
{}

\subsection{Unconditional Speech Generation}
\label{sec:speech}
\label{sec:diffwave}

The SC09 spoken digits dataset 
is 
a challenging unconditional speech generation benchmark, as it contains several axes of variation (words, speakers, microphones, alignments). 
Unlike the music setting (\cref{sec:music}), SC09 contains audio of \emph{bounded} length ($1$-second utterances). %
To date, AR waveform models trained on this benchmark have yet to generate spoken digits which are consistently intelligible to humans.%
\footnote{While AR waveform models can produce intelligible speech in the context of TTS systems, this capability requires conditioning on rich intermediaries like spectrograms or linguistic features.}
In contrast, non-AR approaches are capable of achieving global coherence on this dataset, as first demonstrated by WaveGAN~\citep{donahue2019adversarial}. %

Although our primary focus thus far has been the challenging testbed of AR waveform modeling, 
\name{} can also be used as a flexible neural network architecture for audio generation more broadly. 
We demonstrate this potential by integrating \name{} into DiffWave~\citep{kong2021diffwave}, a diffusion-based method for non-AR waveform generation which represents the current state-of-the-art for SC09. 
DiffWave uses the original WaveNet architecture as its backbone---here, we simply replace WaveNet with a \name{} model containing a similar number of parameters. 

We compare \name{} to strong baselines on SC09 in both the AR and non-AR (via DiffWave) settings by measuring several standard quantitative and qualitative metrics such as Frech\'et Inception Distance (FID) and Inception Score (IS) (\cref{sec:evaluations}). 
We also conduct a qualitative evaluation where we ask several annotators to label the generated digits and then compute their inter-annotator agreement. %
Additionally, as in~\citet{donahue2019adversarial}, we ask annotators for their subjective opinions on overall audio quality, intelligibility, and speaker diversity, and report MOS for each axis. Results for all models appear in \cref{tab:sc09}.

\textbf{Autoregressive.} 
\name{} substantially outperforms other AR waveform models on all metrics, and achieves $2\times$ higher MOS for both quality and intelligibility. 
Moreover, annotators agree on labels for samples from \name{} far more often than they do for samples from other AR models, 
suggesting that \name{} generates waveforms that are more globally coherent on average than prior work.
Finally, \name{} achieves higher MOS on all axes compared to WaveGAN while using more than $4\times$ fewer parameters.

\textbf{Non-autoregressive.}
Integrating \name{} into DiffWave substantially improves performance on all metrics compared to its WaveNet-based counterpart, and achieves a new overall state-of-the-art performance on all quantitative and qualitative metrics on SC09. 
We note that this result involved \emph{zero tuning} of the model or training parameters (e.g.~diffusion steps or optimizer hyperparameters) (\cref{appendix:methods}).
This suggests that \name{} could be useful not only for AR waveform modeling but also as a new drop-in architecture for many audio generation systems which currently depend on WaveNet (see~\cref{sec:related}). 

We additionally conduct several ablation studies on our hybrid DiffWave and \name{} model, and compare performance earlier in training and with smaller models (\cref{tab:sc09-diffusion}). 
When paired with DiffWave, \name{} is much more sample efficient than WaveNet,
matching the performance of the best WaveNet-based model with half as many training steps.
\citet{kong2021diffwave} also observed that DiffWave was extremely sensitive with a WaveNet backbone, performing poorly with smaller models and becoming unstable with larger ones.
We show that, when using WaveNet, a small DiffWave model fails to model the dataset, however it works much more effectively when using \name. 
Finally, we ablate our non-causal relaxation, showing that this bidirectional version of \name{} performs much better than its unidirectional counterpart (as expected).

\section{Discussion}
\label{sec:discussion}

Our results indicate that \name{} is a promising new architecture for modeling raw audio waveforms.
When trained on music and speech datasets,
\name{} generates waveforms that humans judge to be more musical and intelligible respectively compared to waveforms from previous architectures,
indicating that audio generated by \name{} has a higher degree of global coherence.
By leveraging the dual convolutional and recurrent forms of S4,
\name{} is more computationally efficient than past architectures during both training and inference.
Additionally, \name{} is consistently more sample efficient to train---it achieves better quantitative performance with fewer training steps.
Finally, when used as a drop-in replacement for WaveNet, \name{} improved the performance of an existing state-of-the-art model for unconditional generation, indicating a potential for \name{} to create a ripple effect of improving audio generation more broadly.

\section*{Acknowledgments}
We thank John Thickstun for helpful conversations. We gratefully acknowledge the support of NIH under No. U54EB020405 (Mobilize), NSF under Nos. CCF1763315 (Beyond Sparsity), CCF1563078 (Volume to Velocity), and 1937301 (RTML); ARL under No. W911NF-21-2-0251 (Interactive Human-AI Teaming); ONR under No. N000141712266 (Unifying Weak Supervision); ONR N00014-20-1-2480: Understanding and Applying Non-Euclidean Geometry in Machine Learning; N000142012275 (NEPTUNE); NXP, Xilinx, LETI-CEA, Intel, IBM, Microsoft, NEC, Toshiba, TSMC, ARM, Hitachi, BASF, Accenture, Ericsson, Qualcomm, Analog Devices, Google Cloud, Salesforce, Total, the HAI-AWS Cloud Credits for Research program, the Stanford Data Science Initiative (SDSI), and members of the Stanford DAWN project: Facebook, Google, and VMWare. The U.S. Government is authorized to reproduce and distribute reprints for Governmental purposes notwithstanding any copyright notation thereon. Any opinions, findings, and conclusions or recommendations expressed in this material are those of the authors and do not necessarily reflect the views, policies, or endorsements, either expressed or implied, of NIH, ONR, or the U.S. Government. 

\bibliography{biblio}

\newpage
\appendix
\onecolumn

\section{Model Details}
\label{sec:model-details}

\subsection{S4 Stability}

We prove \cref{prop:hippo}.
We build off the S4 representation of HiPPO matrices, using their decomposition as a normal plus low-rank matrix which implies that they are unitarily similar to a diagonal plus low-rank matrix.
Then we show that the low-rank portion of this decomposition is in fact negative semidefinite,
while the diagonal portion has non-positive real part.

\begin{proof}[Proof of \cref{prop:hippo}]%
  We consider the diagonal plus low-rank decompositions shown in \citet{gu2022efficiently} of the three original HiPPO matrices \citet{gu2020hippo},
  and show that the low-rank portions are in fact negative semidefinite.

  \textbf{HiPPO-LagT.}
  The family of generalized HiPPO-LagT matrices are defined by
  \begin{align*}
    \bm{A}_{nk} &=
    \begin{cases}%
      0            & n < k \\
      -\frac{1}{2} - \beta & n=k   \\
      -1           & n > k \\
    \end{cases}
  \end{align*}
  for \( 0 \le \beta \le \frac{1}{2} \), with the main HiPPO-LagT matrix having $\beta=0$.

  It can be decomposed as
  \begin{align*}
    \bm{A} &=
    -
    \begin{bmatrix}
      \frac{1}{2}+\beta &             &             &            & \dots \\
      1           & \frac{1}{2}+\beta &             &             \\
      1           & 1           & \frac{1}{2}+\beta &             \\
      1           & 1           & 1           & \frac{1}{2}+\beta \\
      \vdots      &             &             &              & \ddots \\
    \end{bmatrix}
    =
    - \beta I
    -
    \begin{bmatrix}
      & -\frac{1}{2} & -\frac{1}{2} & -\frac{1}{2} \\
      \frac{1}{2} &              & -\frac{1}{2} & -\frac{1}{2} & \cdots \\
      \frac{1}{2} & \frac{1}{2}  &              & -\frac{1}{2} \\
      \frac{1}{2} & \frac{1}{2}  & \frac{1}{2}  &              \\
      \vdots      &             &             &              & \ddots \\
    \end{bmatrix}
    -
    \frac{1}{2}
    \begin{bmatrix}
      1 & 1 & 1 & 1  & \cdots \\
      1 & 1 & 1 & 1 \\
      1 & 1 & 1 & 1 \\
      1 & 1 & 1 & 1 \\
      \vdots &   &   &    & \ddots \\
    \end{bmatrix}
    .
  \end{align*}
  The first term is skew-symmetric, which is unitarily similar to a (complex) diagonal matrix with pure imaginary eigenvalues (i.e., real part \( 0 \)).
  The second matrix can be factored as \( pp^* \) for \( p = 2^{-1 / 2} \begin{bmatrix} 1 & \cdots & 1 \end{bmatrix}^* \).
  Thus the whole matrix \( A \) is unitarily similar to a matrix \( \Lambda - pp^* \) where the eigenvalues of \( \Lambda \) have real part between \( -\frac{1}{2} \) and \( 0 \).

  \textbf{HiPPO-LegS.}
  The HiPPO-LegS matrix is defined as
  \begin{align*}
    \bm{A}_{nk}
    =
    -
    \begin{cases}
      (2n+1)^{1/2}(2k+1)^{1/2} & \mbox{if } n > k \\
      n+1                      & \mbox{if } n = k \\
      0                        & \mbox{if } n < k
    \end{cases}
    .
  \end{align*}
  It can be decomposed as
  Adding \( \frac{1}{2}(2n+1)^{1/2}(2k+1)^{1/2} \) to the whole matrix gives
  \begin{align*}
    & -\frac{1}{2} I
    - S
    - pp^*
    \\
    S_{nk} &=
    \begin{cases}
      \frac{1}{2} (2n+1)^{1/2}(2k+1)^{1/2}  & \mbox{if } n > k \\
      0                           & \mbox{if } n = k \\
      -\frac{1}{2} (2n+1)^{1/2}(2k+1)^{1/2} & \mbox{if } n < k \\
    \end{cases}
    \\
    p_n &= (n+\frac{1}{2})^{1 / 2}
  \end{align*}
  Note that \( S \) is skew-symmetric.
  Therefore \( A \) is unitarily similar to a matrix \( \Lambda - pp^* \) where the eigenvalues of \( \Lambda \) have real part \( -\frac{1}{2} \).

  \textbf{HiPPO-LegT.}

  Up to the diagonal scaling,
  the LegT matrix is
  \begin{align*}
    \bm{A} =
    -
    \begin{bmatrix}
      1      & -1 & 1  & -1  & \dots \\
      1      & 1  & -1 & 1  \\
      1      & 1  & 1  & -1 \\
      1      & 1  & 1  & 1  \\
      \vdots &    &    &     & \ddots
    \end{bmatrix}
    =
    -
    \begin{bmatrix}
      0 & -1 & 0 & -1 & \cdots \\
      1 & 0 & -1 & 0 \\
      0 & 1 & 0 & -1 \\
      1 & 0 & 1 & 0 \\
      \vdots &    &    &     & \ddots
    \end{bmatrix}
    -
    \begin{bmatrix}
      1 & 0 & 1 & 0 & \cdots \\
      0 & 1 & 0 & 1 \\
      1 & 0 & 1 & 0 \\
      0 & 1 & 0 & 1 \\
      \vdots &    &    &     & \ddots
    \end{bmatrix}
    .
  \end{align*}

  The first term is skew-symmetric and the second term can be written as \( pp^* \) for
  \begin{align*}
    p =
    \begin{bmatrix}
      1 & 0 & 1 & 0 & \cdots \\
      0 & 1 & 0 & 1 & \cdots \\
    \end{bmatrix}^\top
  \end{align*}
\end{proof}

\subsection{Model Architecture}

\paragraph{S4 Block Details}

The first portion of the S4 block is the same as the one used in \citet{gu2022efficiently}.
\begin{align*}
  y &= x \\
  y &= \mathsf{LayerNorm}(y) \\
  y &= \mathsf{S4}(y) \\
  y &= \phi(y) \\
  y &= Wy + b \\
  y &= x + y \\
\end{align*}
Here \( \phi \) is a non-linear activation function, chosen to be GELU \citep{hendrycks2016gaussian} in our implementation.
Note that all operations aside from the S4 layer are \emph{position-wise} (with respect to the time or sequence dimension).

These operations are followed by more position-wise operations, which are standard in other deep neural networks such as Transformers (where it is called the feed-forward network) and CNNs (where it is called the inverted bottleneck layer).
\begin{align*}
  y &= x \\
  y &= \mathsf{LayerNorm}(y) \\
  y &= W_1y + b_1 \\
  y &= \phi(y) \\
  y &= W_2y + b_2 \\
  y &= x + y \\
\end{align*}
Here \( W_1 \in \mathbbm{R}^{d \times ed} \) and \( W_2 \in \mathbbm{R}^{ed \times d} \), where \( e \) is an expansion factor.
We fix \( e=2 \) in all our experiments.

\section{Additional Results}
\label{sec:experiments-additional}
We provide details of ablations, including architecture ablations and efficiency benchmarking.
\subsubsection{YouTubeMix}
 We conduct architectural ablations and efficiency benchmarking for all baselines on the YouTubeMix dataset.
 
{\bf Architectures.} SampleRNN-$2$ and SampleRNN-$3$ correspond to the $2$- and $3$-tier models described in \cref{appendix:methods} respectively. WaveNet-$512$ and WaveNet-$1024$ refer to models with $512$ and $1024$ skip channels respectively with all other details fixed as described in \cref{appendix:methods}. \name-$\{2, 4, 6, 8\}$ consist of the indicated number of S4 blocks in each tier of the architecture, with all other details being the same.

{\bf Isotropic S4.} We also include an isotropic S4 model to ablate the effect of pooling in \name{}. Isotropic S4 can be viewed as \name{} without any pooling (i.e. no additional tiers aside from the top tier). We note that due to larger memory usage for these models, we use a sequence length of $4$s for the $4$ layer isotropic model, and a sequence length of $2$s for the $8$ layer isotropic model (both with batch size $1$), highlighting an additional disadvantage in memory efficiency.

{\bf Throughput Benchmarking.} To measure peak throughput, we track the time taken by models to generate $1000$ samples at batch sizes that vary from $1$ to $8192$ in powers of $2$. The throughput is the total number of samples generated by a model in $1$ second. Figure~\ref{fig:throughput-all} shows the results of this study in more detail for each method.

\begin{figure*}[t]
    \centering
    \includegraphics[width=\linewidth]{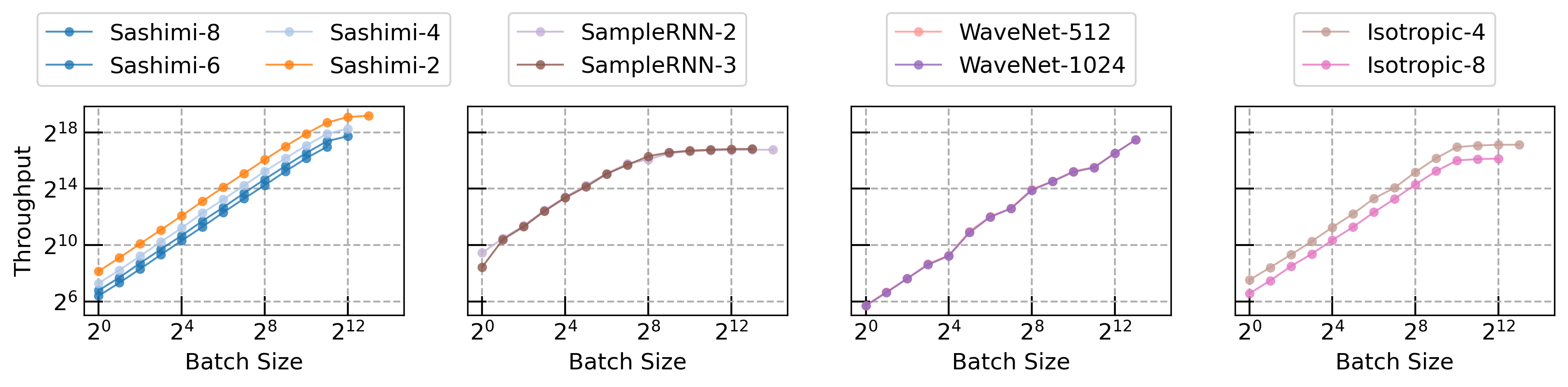}
    \caption{Log-log plot of throughput vs. batch size. Throughput scales near linearly for \name. By contrast, SampleRNN throughput peaks at smaller batch sizes, while WaveNet shows sublinear scaling with throughput degradation at some batch sizes. Isotropic variants have far lower throughput than \name.}
    \label{fig:throughput-all}
\end{figure*}

\begin{figure}[t]
    \centering
    \includegraphics[width=0.5\linewidth]{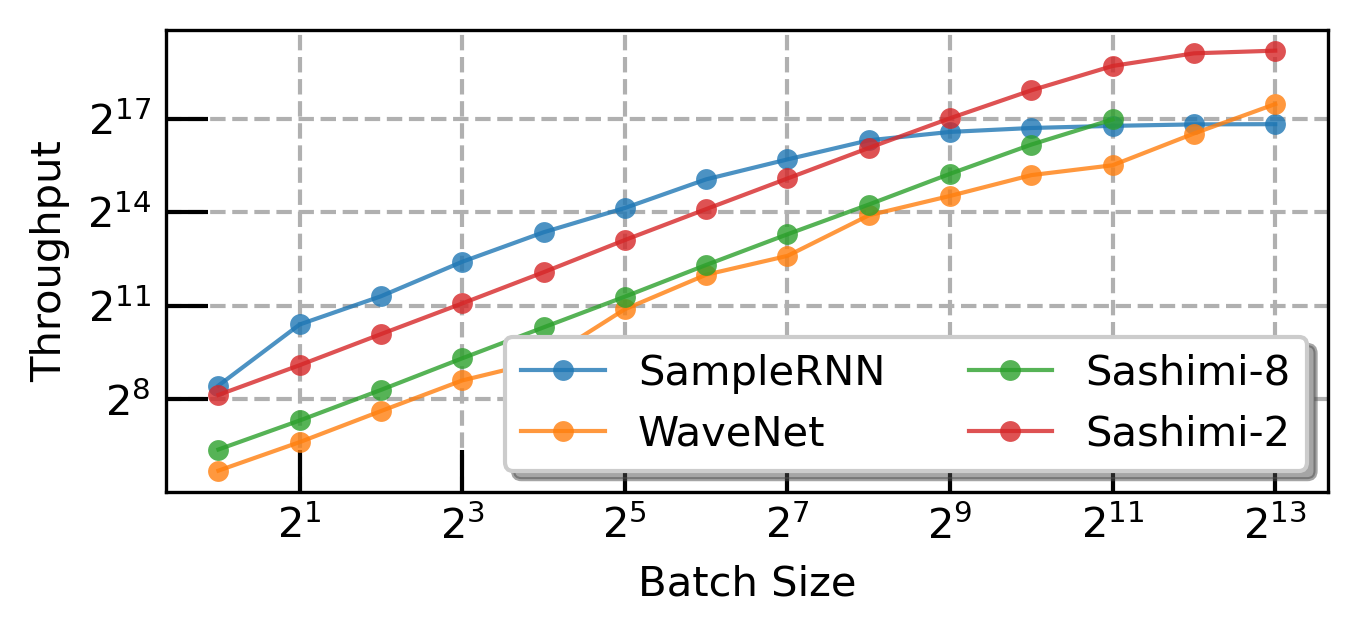}
    \caption{Log-log plot of throughput vs. batch size. \name{}-$2$ improves peak throughput over WaveNet and SampleRNN by $3\times$ and $5\times$ respectively.}
    \label{fig:throughput}
\end{figure}

{\bf Diffusion model ablations.}
\cref{tab:sc09-diffusion} reports results for the ablations described in \cref{sec:diffwave}.
Experimental details are provided in \cref{appendix:methods}.

\iftoggle{arxiv}
{}
{
    \begin{table*}[t]
        \centering
        \caption{
          (\textbf{SC09 diffusion models}.)
          Beyond AR, \name{} can be flexibly combined with other generative modeling approaches, improving on the state-of-the-art DiffWave model by simply replacing the architecture.
          (\emph{Top}) A parameter-matched \name{} architecture with \emph{no tuning} outperforms the best DiffWave model.
          (\emph{Middle}) \name{} is consistently better than WaveNet at all stages of training; a model trained on half the samples matches the best DiffWave model.
          (\emph{Bottom}) The WaveNet backbone is extremely sensitive to architecture parameters such as size and dilation schedule; a small model fails to learn. We also ablate the bidirectional S4 layer, which outperforms the unidirectional one.
        }
        \sc
        \resizebox{\linewidth}{!}{
            \begin{tabular}{llllllllllll}
                \toprule
                            \multirow{2}{*}{Architecture} & \multirow{2}{*}{Params} & \multirow{2}{*}{Training Steps} & \multirow{2}{*}{FID $\downarrow$} & \multirow{2}{*}{IS $\uparrow$} & \multirow{2}{*}{mIS $\uparrow$} & \multirow{2}{*}{AM $\downarrow$} & \multirow{2}{*}{NDB $\downarrow$} & \multirow{2}{*}{\shortstack{Human (\( \kappa \))\\Agreement}} & \multicolumn{3}{c}{MOS} \\
                \cmidrule{10-12}
                &                         &                                   &                                &                                 &                                  &                                   & & & Quality                & Intelligibility     & Diversity          \\
    
                \midrule
                \name{}        & 23.0M      & 800k           & $1.42$              & $5.94$           & $69.17$            & $0.59$             & $0.88$              &   $0.953$                             &  $4.20 \pm 0.06$               & $4.33 \pm 0.03$ & $3.28 \pm 0.11$                      \\
                WaveNet        & 24.1M      & 1000k          & $1.92$           & $5.26$        & $51.21$        & $0.68$          & $0.88$           & $0.917$                        & $4.03 \pm 0.06$ & $4.15 \pm 0.03$        & $3.45 \pm 0.09$ \\
                \midrule
                \name{}        & 23.0M      & 500k           & $2.08$           & $5.68$        & $51.10$        & $0.66$          & $0.76$           & $0.923$                        & $3.99 \pm 0.06$ & $4.13 \pm 0.03$        & $3.38 \pm 0.10$ \\
                WaveNet        & 24.1M      & 500k           & $2.25$           & $4.68$        & $34.55$        & $0.80$           & $0.90$            & $0.848$                        & $3.53 \pm 0.07$ & $3.69 \pm 0.03$        & $3.30 \pm 0.08$ \\
                \midrule
                \name{} (uni.) & 7.1M       & 500k           & $2.70$              & $3.62$           & $17.96$            & $1.03$             & $0.90$              &  $0.829$                              & $3.08 \pm 0.07$                 & $3.29 \pm 0.04$ & $3.26 \pm 0.08$                       \\
                \name{}        & 7.5M       & 500k           & $1.70$              & $5.00$           & $40.27$        & $0.72$             & $0.90$              & $0.934$                               &  $3.83 \pm 0.07$               &     $4.00 \pm 0.03$ & $3.34 \pm 0.09$                   \\
                WaveNet        & 6.8M       & 500k           & $4.53$              & $2.80$           & $9.02$        & $1.30$             & $0.94$              &   $0.446$                             &  $1.85 \pm 0.08$               & $1.90 \pm 0.03$ & $3.03 \pm 0.10$                      \\
                \bottomrule
            \end{tabular}
        }
        \label{tab:sc09-diffusion}
    \end{table*}
}

\section{Experiment Details}
We include experimental details, including dataset preparation, hyperparameters for all methods, details of ablations as well as descriptions of automated and human evaluation metrics below.
\subsection{Datasets}

A summary of dataset information can be found in Table~\ref{tab:datasets}. Across all datasets, audio waveforms are preprocessed to $16$kHz using \texttt{torchaudio}.

{\bf Beethoven.} The dataset consists of recordings of Beethoven's $32$ piano sonatas. We use the version of the dataset shared by \citet{mehri2017samplernn}, which can be found \href{https://drive.google.com/drive/folders/0B7riq_C8aslvbWJuMGhJRFBmSHM?resourcekey=0-fM79ZaHDzE4IPUMzDUK6uA}{here}. Since we compare to numbers reported by \citet{mehri2017samplernn}, we use linear quantization for all (and only) Beethoven experiments. We attempt to match the splits used by the original paper by reference to the code provided \href{https://github.com/soroushmehr/sampleRNN_ICLR2017}{here}.

{\bf YouTubeMix.} A $4$ hour dataset of piano music taken from \url{https://www.youtube.com/watch?v=EhO_MrRfftU}. We split the audio track into \texttt{.wav} files of $1$ minute each, and use the first $88\%$ files for training, next $6\%$ files for validation and final $6\%$ files for testing.

{\bf SC09.}
The Speech Commands dataset~\cite{Warden2018SpeechCA} contains many spoken words by thousands of speakers under various recording conditions including some very noisy environments.
Following prior work~\citep{donahue2019adversarial,kong2021diffwave} we use the subset that contains spoken digits ``zero'' through ``nine''.
This SC09 dataset contains 31,158 training utterances (8.7 hours in total) by 2,032 speakers, where each audio
has length 1 second sampled at 16kHz.
the generative models need to model them without any conditional information.

The datasets we used can be found on Huggingface datasets: \href{https://huggingface.co/datasets/krandiash/beethoven}{Beethoven}, \href{https://huggingface.co/datasets/krandiash/youtubemix}{YouTubeMix}, \href{https://huggingface.co/datasets/krandiash/sc09}{SC09}.

\subsection{Models and Training Details}
\label{appendix:methods}

For all datasets, \name{}, SampleRNN and WaveNet receive $8$-bit quantized inputs. During training, we use no additional data augmentation of any kind. We summarize the hyperparameters used and any sweeps performed for each method below. 

\subsubsection{Details for Autoregressive Models}
\label{appendix:methods-autoregressive}
All methods in the AR setting were trained on single V100 GPU machines.

\paragraph{\name.}
We adapt the S4 implementation provided by \citet{gu2022efficiently} to incorporate parameter tying for $pq^*$. For simplicity, we do not train the low-rank term $pp^*$, timescale $dt$ and the $B$ matrix throughout our experiments, and let $\diag$ be trained freely. We find that this is actually stable, but leads to a small degradation in performance compared to the original S4 parameterization. Rerunning all experiments with our updated Hurwitz parameterization--which constrains the real part of the entries of $\diag$ using an $\exp$ function--would be expensive, but would improve performance. For all datasets, we use feature expansion of $2\times$ when pooling, and use a feedforward dimension of $2\times$ the model dimension in all inverted bottlenecks in the model. We use a model dimension of $64$. For S4 parameters, we only train $\Lambda$ and $C$ with the recommended learning rate of $0.001$, and freeze all other parameters for simplicity (including $pp^*, B, dt$). We train with $4\times \rightarrow 4\times$ pooling for all datasets, with $8$ S4 blocks per tier. 

On Beethoven, we learn separate $\Lambda$ matrices for each SSM in the S4 block, while we use parameter tying for $\Lambda$ within an S4 block on the other datasets. On SC09, we found that swapping in a gated linear unit (GLU)~\citep{dauphin2017language} in the S4 block improved NLL as well as sample quality.

We train \name{} on Beethoven for $1$M steps, YouTubeMix for $600$K steps, SC09 for $1.1$M steps.

\paragraph{SampleRNN.}
We adapt an \href{https://github.com/deepsound-project/samplernn-pytorch}{open-source PyTorch implementation} of the SampleRNN backbone, and train it using truncated backpropagation through time (TBPTT) with a chunk size of $1024$. We train $2$ variants of SampleRNN: a 3-tier model with frame sizes $8, 2, 2$ with $1$ RNN per layer to match the 3-tier RNN from~\citet{mehri2017samplernn} and a 2-tier model with frame sizes $16, 4$ with $2$ RNNs per layer that we found had stronger performance in our replication (than the 2-tier model from~\citet{mehri2017samplernn}). For the recurrent layer, we use a standard GRU model with orthogonal weight initialization following~\citet{mehri2017samplernn}, with hidden dimension $1024$ and feedforward dimension $256$ between tiers. We also use weight normalization as recommended by~\citet{mehri2017samplernn}.

We train SampleRNN on Beethoven for $150$K steps, YouTubeMix for $200$K steps, SC09 for $300$K steps. We found that SampleRNN could be quite unstable, improving steadily and then suddenly diverging. It also appeared to be better suited to training with linear rather than mu-law quantization.

\paragraph{WaveNet.}
We adapt an \href{https://github.com/vincentherrmann/pytorch-wavenet}{open-source PyTorch implementation} of the WaveNet backbone, trained using standard backpropagation. We set the number of residual channels to $64$, dilation channels to $64$, end channels to $512$. We use $4$ blocks of dilation with $10$ layers each, with a kernel size of $2$. Across all datasets, we sweep the number of skip channels among $\{512, 1024\}$. For optimization, we use the AdamW optimizer, with a learning rate of $0.001$ and a plateau learning rate scheduler that has a patience of $5$ on the validation NLL. During training, we use a batch size of $1$ and pad each batch on the left with zeros equal to the size of the receptive field of the WaveNet model ($4093$ in our case).

We train WaveNet on Beethoven for $400$K steps, YouTubeMix for $200$K steps, SC09 for $500$K steps.

\subsubsection{Details for Diffusion Models}
\label{appendix:methods-diffusion}
All diffusion models were trained on 8-GPU A100 machines.

\paragraph{DiffWave.}
We adapt an \href{https://github.com/philsyn/DiffWave-unconditional}{open-source PyTorch implementation} of the DiffWave model.
The DiffWave baseline in \cref{tab:sc09} is the unconditional SC09 model reported in \citet{kong2021diffwave},
which uses a 36 layer WaveNet backbone with dilation cycle \( [1, 2, 4, 8, 16, 32, 64, 128, 256, 512, 1024, 2048] \) and hidden dimension \( 256 \),
a linear diffusion schedule \( \beta_t \in [1 \times 10^4, 0.02] \) with \( T=200 \) steps,
and the Adam optimizer with learning rate \( 2\times 10^{-4} \).
The small DiffWave model reported in \cref{tab:sc09-diffusion} has 30 layers with dilation cycle \( [1, 2, 4, 8, 16, 32, 64, 128, 256, 512] \) and hidden dimension \( 128 \).

\paragraph{DiffWave with \name.}
Our large \name{} model has hidden dimension \( 128 \) and \( 6 \) S4 blocks per tier with the standard two pooling layers with pooling factor \( 4 \) and expansion factor \( 2 \) (\cref{sec:architecture}).
We additionally have S4 layers in the down-blocks in addition to the up-blocks of \cref{fig:architecture}.
Our small \name{} model (\cref{tab:sc09-diffusion}) reduces the hidden dimension to \( 64 \).
These architectures were chosen to roughly parameter match the DiffWave model.
While DiffWave experimented with other architectures such as deep and thin WaveNets or different dilation cycles~\citep{kong2021diffwave},
we only ran a single \name{} model of each size.
All optimization and diffusion hyperparameters were kept the same,
with the exception that we manually decayed the learning rate of the large \name{} model at $500$K steps as it had saturated and the model had already caught up to the best DiffWave model (\cref{tab:sc09-diffusion}).

\subsection{Automated Evaluations}
\label{sec:evaluations}

{\bf NLL.} We report negative log-likelihood (NLL) scores for all AR models in bits, on the test set of the respective datasets. To evaluate NLL, we follow the same protocol as we would for training, splitting the data into non-overlapping chunks (with the same length as training), running each chunk through a model and then using the predictions made on each step of that chunk to calculate the average NLL for the chunk.

{\bf Evaluation of generated samples.}
Following \citet{kong2021diffwave}, we use 4 standard evaluation metrics for generative models evaluated using an auxiliary ResNeXT classifier \citep{xie2017aggregated} which achieved 98.3\% accuracy on the test set. Note that \citet{kong2021diffwave} reported an additional metric NDB (number of statistically-distinct bins), which we found to be slow to compute and generally uninformative, despite \name{} performing best.
\begin{itemize}%
  \item \textbf{Fr\'echet Inception Distance (FID)} \citep{heusel2017gans} uses the classifier to compare moments of generated and real samples in feature space.
  \item \textbf{Inception Score (IS)} \citep{salimans2016improved} measures both quality and diversity of generated samples, and favoring samples that the classifier is confident on.
  \item \textbf{Modified Inception Score (mIS)} \citep{gurumurthy2017deligan}
  provides a measure of both intra-class in addition to inter-class diversity.
  \item \textbf{AM Score} \citep{zhou2018activation} uses the marginal label distribution of training data compared to IS.
\end{itemize}

We also report the Cohen's inter-annotator agreement $\kappa$ score, which is computed with the classifier as one rater and a crowdworker's digit prediction as the other rater (treating the set of crowdworkers as a single rater).

\subsubsection{Evaluation Procedure for Autoregressive Models}

Because autoregressive models have tractable likelihood scores that are easily evaluated,
we use them to perform a form of rejection sampling when evaluating their automated metrics.
Each model generated \( 5120 \) samples and ranked them by likelihood scores.
The lowest-scoring $0.40$ and highest-scoring $0.05$ fraction of samples were thrown out.
The remaining samples were used to calculate the automated metrics.

The two thresholds for the low- and high- cutoffs were found by validation on a separate set of $5120$ generated samples.

\subsubsection{Evaluation Procedure for Non-autoregressive Models}

Automated metrics were calculated on 2048 random samples generated from each model.

\subsection{Evaluation of Mean Opinion Scores}
\label{sec:mos}

For evaluating mean opinion scores (MOS), we repurpose scripts for creating jobs for Amazon Mechanical Turk from \citet{Neekhara2019ExpeditingTS}.

\begin{figure}
    \centering
    \caption{{\bf (YouTubeMix MOS Interface)} Crowdsourcing interface for collecting mean opinion scores (MOS) on YouTubeMix. Crowdworkers are given a collection of audio files, one from each method and the dataset. They are asked to rate each file on audio fidelity and musicality.}
    \includegraphics[width=\linewidth]{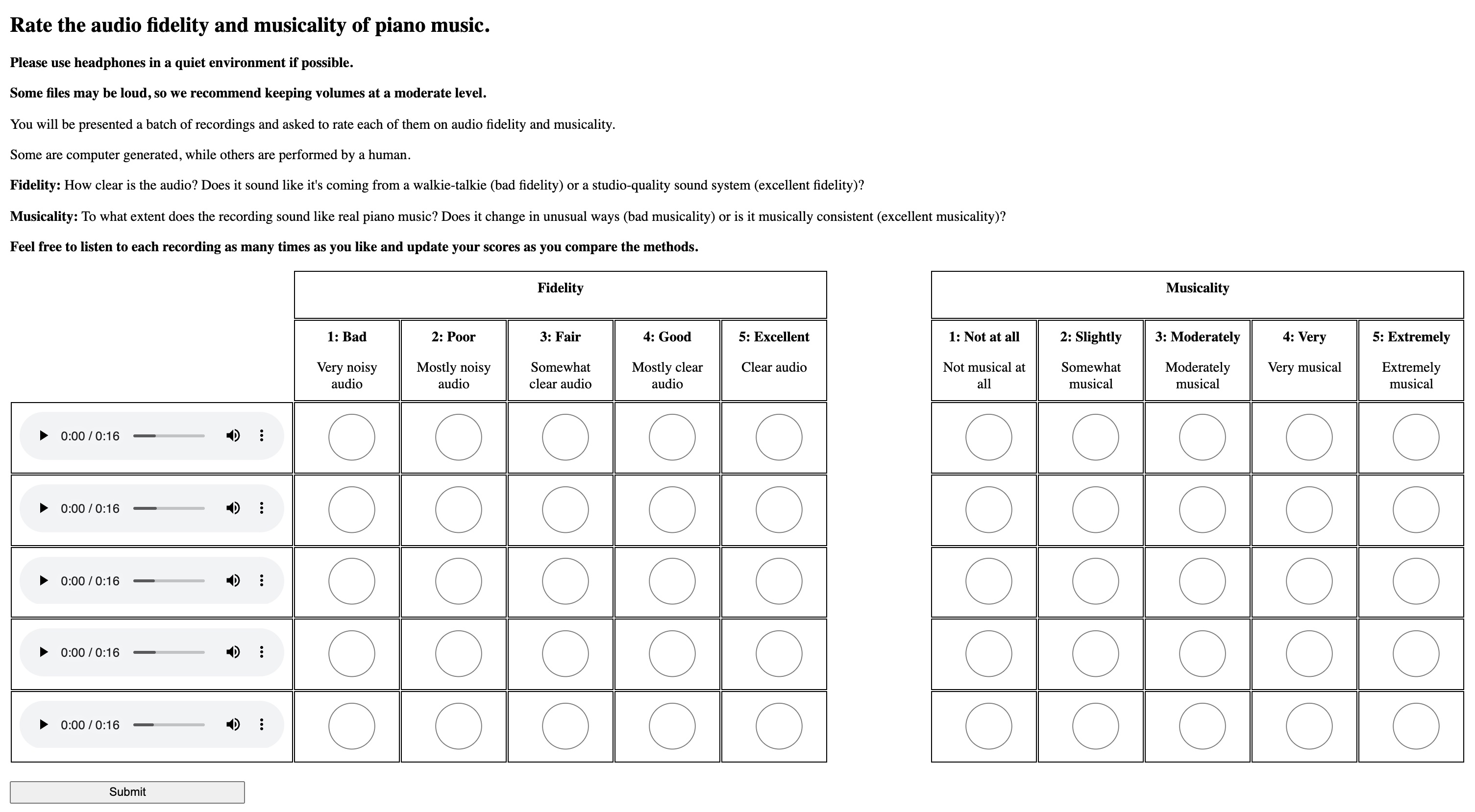}
    \label{fig:mturk-interface-music}
\end{figure}

\begin{figure}
    \centering
    \caption{{\bf (SC09 MOS Interface)} Crowdsourcing interface for collecting mean opinion scores (MOS) on SC09. Crowdworkers are given a collection of $10$ audio files from the same method, and are asked to classify the spoken digits and rate them on intelligibility. At the bottom, they provide a single score on the audio quality and speaker diversity they perceive for the batch.}
    \begin{subfigure}{\linewidth}
        \includegraphics[width=\linewidth]{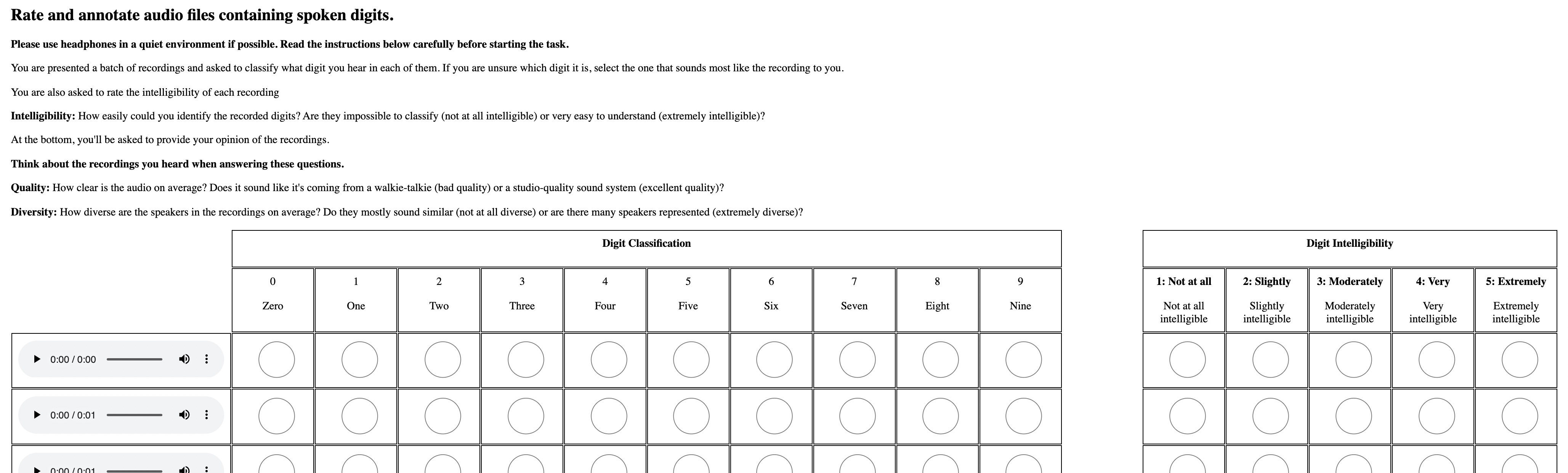}
    \end{subfigure}
    \begin{subfigure}{\linewidth}
        \includegraphics[width=\linewidth]{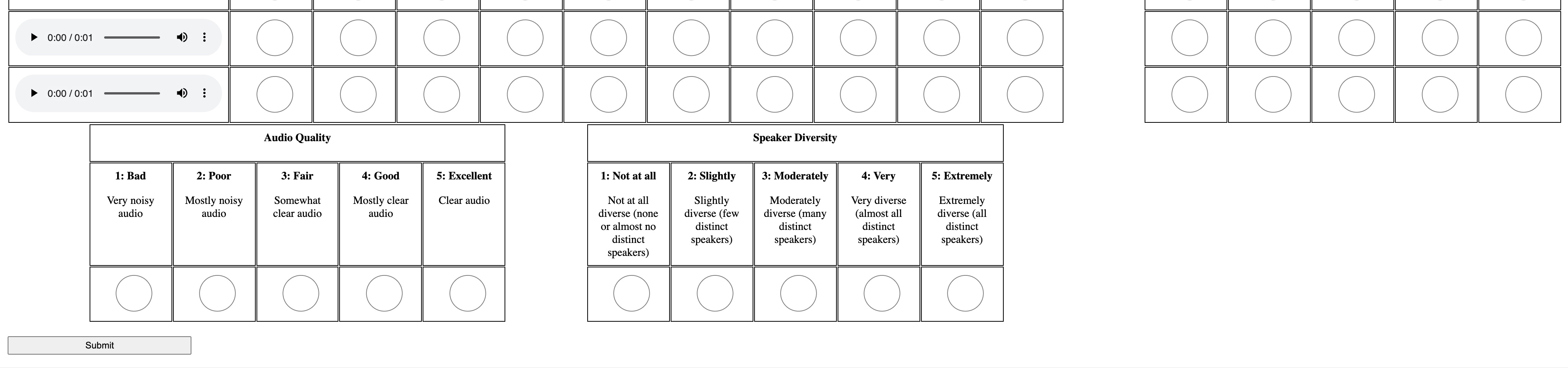}
    \end{subfigure}
    \label{fig:mturk-interface-digits}
\end{figure}

\subsubsection{Mean Opinion Scores for YouTubeMix}
We collect MOS scores on audio fidelity and musicality, following \citet{dieleman2018challenge}. The instructions and interface used are shown in Figure~\ref{fig:mturk-interface-music}.

The protocol we follow to collect MOS scores for YouTubeMix is outlined below. For this study, we compare unconditional AR models, \name{} to SampleRNN and WaveNet.
\begin{itemize}
    \item For each method, we generated unconditional $1024$ samples, where each sample had length $16$s ($1.024$M steps). For sampling, we directly sample from the distribution output by the model at each time step, without using any other modifications.
    \item As noted by \citet{mehri2017samplernn}, autoregressive models can sometimes generate samples that are ``noise-like". To fairly compare all methods, we sequentially inspect the samples and reject any that are noise-like. We also remove samples that mostly consist of silences (defined as more than half the clip being silence). We carry out this process until we have $30$ samples per method.
    \item Next, we randomly sample $25$ clips from the dataset. Since this evaluation is quite subjective, we include some gold standard samples. We add $4$ clips that consist mostly of noise (and should have musicality and quality MOS $<=2$). We include $1$ clip that has variable quality but musicality MOS $<=2$. Any workers who disagree with this assessment have their responses omitted from the final evaluation.
    \item We construct $30$ batches, where each batch consists of $1$ sample per method (plus a single sample for the dataset), presented in random order to a crowdworker. We use Amazon Mechanical Turk for collecting responses, paying $\$0.50$ per batch and collecting $20$ responses per batch. We use Master qualifications for workers, and restrict to workers with a HIT approval rating above $98\%$. We note that it is likely enough to collect $10$ responses per batch.
\end{itemize}

\subsubsection{Mean Opinion Scores for SC09}
Next, we outline the protocol used for collecting MOS scores on SC09. We collect MOS scores on digit intelligibility, audio quality and speaker diversity, as well as asking crowdworkers to classify digits following \citet{donahue2019adversarial}. The instructions and interface used are shown in Figure~\ref{fig:mturk-interface-digits}.
\begin{itemize}
    \item For each method, we generate $2048$ samples of $1$s each. For autoregressive models (\name{}, SampleRNN, WaveNet), we directly sample from the distribution output by the model at each time step, without any modification. For WaveGAN, we obtained $50000$ randomly generated samples from the authors, and subsampled $2048$ samples randomly from this set. For the diffusion models, we run $200$ steps of denoising following \citet{kong2021diffwave}.
    \item We use the ResNeXT model (\cref{sec:evaluations}) to classify the generated samples into digit categories. Within each digit category, we choose the top-$50$ samples, as ranked by classifier confidence. We note that this mimics the protocol followed by \citet{donahue2019adversarial}, which we established through correspondence with the authors.
    \item Next, we construct batches consisting of $10$ random samples (randomized over all digits) drawn from a single method (or the dataset). Each method (and the dataset) thus has $50$ total batches. We use Amazon Mechanical Turk for collecting responses, paying $\$0.20$ per batch and collecting $10$ responses per batch. We use Master qualification for workers, and restrict to workers with a HIT approval rating above $98\%$.
\end{itemize}
Note that we elicit digit classes and digit intelligibility scores for each audio file, while audio quality and speaker diversity are elicited once per batch.

\end{document}